\documentclass[a4paper,UKenglish]{article}

\usepackage{amssymb}
\usepackage{amsmath}
\usepackage{amsthm}
\usepackage{thm-restate}
\usepackage{hyperref}
\usepackage{cleveref}
\usepackage{caption}
\usepackage{mathtools}
\usepackage{wrapfig}
\usepackage[linesnumbered,ruled,vlined]{algorithm2e}
\SetKwInput{KwInput}{Input}               
\SetKwInput{KwOutput}{Output}

\title{An Instance-optimal Algorithm for Bichromatic Rectangular Visibility}

\author{Jean Cardinal \thanks{Department of Computer Science, Université libre de Bruxelles, {\tt jcardin@ulb.ac.be}}
\and 
Justin Dallant \thanks{Department of Computer Science, Université libre de Bruxelles, {\tt justin.dallant@ulb.be} Supported by the French Community of Belgium via the funding of a FRIA grant.}
\and 
John Iacono \thanks{Department of Computer Science, Université libre de Bruxelles, {\tt john@johniacono.com} Supported by the Fonds de la Recherche Scientifique-FNRS under Grant no MISU F 6001 1.}
}

\date{}

\newcommand{\N}{\mathbb{N}}

\newcommand{\R}{\mathbb{R}}

\newcommand{\bigO}{\mathcal{O}}
\newcommand{\A}{\mathcal{A}}
\newcommand{\entropy}{{\ensuremath\cal H}}

\newcommand{\OPT}{\mbox{\rm OPT}}
\newcommand{\NE}{\mbox{\rm NE}}
\newcommand{\NW}{\mbox{\rm NW}}
\newcommand{\SE}{\mbox{\rm SE}}
\newcommand{\SW}{\mbox{\rm SW}}
\newcommand{\cross}{\mbox{\rm cross}}

\newtheorem{theorem}{Theorem}
\newtheorem{problem}[theorem]{Problem}
\newtheorem{observation}[theorem]{Observation}
\newtheorem{definition}[theorem]{Definition}
\newtheorem{lemma}[theorem]{Lemma}
\newtheorem{corollary}[theorem]{Corollary}
\newtheorem{proposition}[theorem]{Proposition}

\DeclarePairedDelimiter\ceil{\lceil}{\rceil}
\DeclarePairedDelimiter\floor{\lfloor}{\rfloor}

\renewcommand{\epsilon}{\ensuremath\varepsilon}

\renewcommand{\phi}{\ensuremath{\varphi}}

\begin{document}

\maketitle

\begin{abstract}
Afshani, Barbay and Chan (2017) introduced the notion of instance-optimal algorithm in the order-oblivious setting. An algorithm $A$ is instance-optimal in the order-oblivious setting for a certain class of algorithms $\A$ if the following hold: 
\begin{itemize}
    \item $A$ takes as input a sequence of objects from some domain;
    \item for any instance $\sigma$ and any algorithm $A'\in \A$, the runtime of $A$ on $\sigma$ is at most a constant factor removed from the runtime of $A'$ on the worst possible permutation of $\sigma$.
\end{itemize}
If we identify permutations of a sequence as representing the same instance, this essentially states that $A$ is optimal on every possible input (and not only in the worst case).

We design instance-optimal algorithms for the problem of reporting, given a bichromatic set of points in the plane $S$, all pairs consisting of points of different color which span an empty axis-aligned rectangle (or reporting all points which appear in such a pair). This problem has applications for training-set reduction in nearest-neighbour classifiers. It is also related to the problem consisting of finding the decision boundaries of a euclidean nearest-neighbour classifier, for which Bremner et al.\ (2005) gave an optimal output-sensitive algorithm.

By showing the existence of an instance-optimal algorithm in the order-oblivious setting for this problem we push the methods of Afshani et al.\ closer to their limits by adapting and extending them to a setting which exhibits highly non-local features. Previous problems for which instance-optimal algorithms were proven to exist were based solely on local relationships between points in a set.

\end{abstract}
\vfill
\pagebreak

\section{Introduction}

In the theoretical study of algorithms one often quantifies the performance of an algorithm in terms of the worst case or average running time over a distribution of inputs of a given size. Sometimes, more precise statements can be made about the speed of an algorithm on certain instances by expressing the running time in terms of some parameter depending on the input. One class of such algorithms are the so-called output-sensitive algorithms, where the parameter is the size of the output. In computational geometry, a classical example is computing the convex hull of a set of $n$ points in the plane in $\bigO(n\log h)$ time, where $h$ is the size of the convex hull \cite{Kirkpatrick1986, Chan1996-CH}. 
More recently, Afshani et al.~\cite{Afshani-2017} introduced a specific notion of \emph{instance optimality in the order-oblivious setting} and designed algorithms with this property for different problems on point sets. Roughly speaking, an algorithm is instance optimal in a certain class of algorithms $\A$ if on any input sequence $S$, its performance on $S$ is at most a constant factor removed from the best performance of any algorithm in $\A$ on $S$. In the order-oblivious setting, performance is defined as the worst-case runtime over all permutations of the input sequence (the motivation behind this will be made clear below). A common characteristic of most problems solved by Afshani et al.\ is that they are based around local relations between points in $S$, in the sense that the relation between two points $p,q\in S$ depends solely on their coordinates and not on those of any other point in $S$. This is important because it allows one to decompose certain queries into independent queries on a partition of $S$. In this paper we push these methods closer to their limits by adapting and applying them to a problem which does not fit in this framework.

\noindent
\textbf{The studied problem}\quad
Here we consider the following problem: given a set $S$ of points $n$ in the plane with distinct $x$ and $y$ coordinates, each colored red or blue, report the red/blue pairs of points such that the rectangles they span contain no point of $S$ in their interior. This again has applications to machine learning and more specifically nearest-neighbour classifiers. Indeed, by solving this problem and discarding all points which do not not appear in such a pair, we obtain a (possibly much smaller) set of points where for any point $y$ in the plane, the color of its nearest neighbour is the same as in the original set of points for the $L_1$ distance (as pointed out in a paper by Ichino and Sklansky~\cite{Ichino1985}, where this problem was perhaps first studied). This remains true even when \textit{a priori} unknown (and possibly non-linear) scalings might be applied to the $x$ and $y$ axis after this preprocessing step. In fact, the resulting set of points consists of exactly those necessary for this to hold, so this constitutes an optimal reduction of the training set in that sense. Note that this problem does not fit in the general framework of Afshani et al., as whether two point span an empty rectangle depends on the position of all other points. 

\noindent
\textbf{Related works}\quad
Pairs of points spanning empty rectangles and the corresponding graphs have been studied at numerous occasions in the past, under various names. They are called \emph{rectangular influence graphs} by Ichino and Sklansky~\cite{Ichino1985}, who discuss applications to data clustering and classification. In a paper by G\"uting et al.~\cite{Guting-1989} a similar relation is called \emph{direct dominance}, and a worst-case optimal algorithm to report all pairs of related points is given. This algorithm runs in $O(n\log n + h)$ time, where $h$ is the size of the output. A straightforward adaptation yields an algorithm with the same running time for the bichromatic problem studied here. In a paper by Overmars and Wood~\cite{Overmars1988} this relation is called \emph{rectangular visibility} and a different algorithm with the same running time is given as well as algorithms for the dynamic query setting. The expected size of a largest independent set in this graph is studied by Chen et al.~\cite{Chen2009} (where they call such graphs \emph{Delaunay graphs with respect to rectangles}). Generalizations and variations of this type of relation between points have also been widely studied \cite{Overmars1988Conn, Munro1987, Cardinal2009, Agarwal1992, Jaromczyk1992, Kirkpatrick1985, Keil1992, Yao1982, Devillers2008}.
Another problem of note which is closely related to the one we study here is the following: given a set $S$ of $n$ points in the plane, each colored red or blue, compute the subset of edges of the Vorono\"i diagram of $S$ which are adjacent to both a site corresponding to a blue point and a site corresponding to a red point. This problem has some relevance to machine learning as we can equivalently state it as finding the boundaries of a nearest-neighbour classifier with two classes in the plane. A third formulation is finding the pairs of red and blue points such that there is an empty disk whose boundary passes through both. Bremner et al.~\cite{Bremner2005} show that this problem can be solved in output-sensitive optimal $\bigO(n\log h)$ time, where $h$ is the size of the output. It is an interesting open problem to find instance-optimal algorithms for this problem in the order-oblivious setting (or prove that no such algorithm exists). 

\noindent
\textbf{Paper organization}\quad
In Section 2 we motivate and state more precisely the notion of instance-optimality we work with in this paper. In Section 3 we define the problem formally and give an instance lower bound in the order-oblivious model by adapting the adversarial argument of Afshani et al.~\cite{Afshani-2017}. The key new ingredients are a new definition of safety and a way to deal with the fact that here the adversary cannot necessarily change the expected output of the algorithm by moving a single point inside a so-called non-safe region. In Section 4 we give an algorithm and prove that its runtime matches the lower bound. The main observation which makes this work is that while the algorithms by Afshani et al.\ require the safety queries to be decomposable (which they are not here), we can afford to do some preprocessing to make them behave as if they were decomposable, as long as the amount of work done stays within a constant of the lower bound. In section 5 we mention that when competing against algorithms which can do linear queries, instance-optimality in the order-oblivious setting is impossible.

\section{Instance optimality in the order-oblivious setting and model of computation}\label{sec:instance-optimality}
Ideally, we would like to consider a very strong notion of optimality, where an algorithm is optimal if on every instance its runtime is at most a constant factor removed from the algorithm with the smallest runtime on that particular instance. There is an obvious flaw with this definition, as for every instance we can have an algorithm ``specialized'' for that instance, which simply checks if the input is the one it is specialized for then returns the expected output without any further computation when it is the case (and spends however much time it needs to compute the correct output otherwise). For problems which are not solvable in linear time in the worst-case, this prohibits the existence of such optimal algorithms.
One way to get around this issue in some cases and get a meaningful notion of instance-optimality is the following, introduced by Afshani et al.~\cite{Afshani-2017}.

\begin{definition}
Consider a problem where the input consists of a sequence of $n$ elements from a domain ${\cal D}$.  Consider a class $\A$ of algorithms.
A \emph{correct} algorithm refers to an algorithm that outputs a correct answer for every possible sequence of elements in~${\cal D}$.
For a set $S$ of $n$ elements in ${\cal D}$, let $T_A(S)$ denote the maximum running time of $A$ on input $\sigma$ over all $n!$ possible permutations $\sigma$ of $S$.
Let $\OPT(S)$ denote the minimum of $T_{A'}(S)$ over all correct algorithms $A'\in\A$.  If $A\in\A$ is a correct algorithm such that $T_A(S)\le O(1)\cdot\OPT(S)$ for every set $S$, then we say $A$ is \emph{instance-optimal in the order-oblivious setting}.
\end{definition}

By measuring the performance of an algorithm on an instance as the maximum runtime over all permutations of the instance elements, the algorithm can no longer take advantage of the order in which the input elements are presented. In particular, simply checking if the input is a specific sequence is no longer a good strategy.
Here, the domain $\mathcal{D}$ consists of all points in the plane, colored red or blue. An instance is a sequence of points, no two sharing the same $x$ or $y$ coordinate. However, we really want to consider this sequence as a set of points, as the order in which the points are presented does not change the instance conceptually. Thus, it makes sense for us to consider a performance metric for which algorithms cannot take advantage of this order. For the class of algorithms $\A$, we will consider algorithms in a restricted real RAM model where the input can only be accessed through comparison queries. That is, the algorithms can compare the $x$ or $y$ coordinates of two points but not, for example, evaluate arbitrary arithmetic expressions on these coordinates. We refer to such algorithms as \emph{comparison-based algorithms}. The lower bound works even for a stronger model of computation, comparison-based decision trees (assuming at least a unit cost for every point returned in the output). We could also allow the comparison of the $x$ coordinate of a point with the $y$ coordinate of another without changing any of the results of this paper.

\section{Lower-bound for comparison-based algorithms in the order-oblivious setting}\label{sec:lower_bound}
\paragraph*{Some basic notations and definitions}
Throughout this section we consider a set $S$ of $n$ red and blue points in the plane. We assume that $S$ is non-degenerate, in the sense that no two points in $S$ share the same $x$ or $y$ coordinate (in particular, all points are distinct). If $p$ is a point, we will denote its $x$ and $y$ coordinates as $x(p)$ and $y(p)$ respectively and its color as $c(p)$.

\begin{definition}
A point $p$ \emph{dominates} $q\neq p$ if $x(p) \geq x(q)$ and $y(p) \geq y(q)$.
A point is \emph{maximal} (resp.\ \emph{minimal}) in $S$ if is dominated by (resp.\ dominates) no point in $S$.
\end{definition}

\begin{definition}[Visible and participating points]
Let $p,q$ be two points in $S$. We say that $q$ is \emph{visible} from $p$ in $S$ (or that $p$ \emph{sees} $q$ in $S$) if the axis-aligned box spanned by $p$ and $q$ contains no point of $S$ in its interior.
We say that $p\in S$ \emph{participates} in $S$ if it is visible from a point in $S$ of the opposite color.
We will omit the set $S$ when it is clear from context.
\end{definition}

The problems we want to solve can thus be restated as follows.
\begin{problem}[Reporting participating points]
Report all points which participate in $S$.
\end{problem}

\begin{problem}[Reporting red-blue pairs of visible points]
Report all red-blue pairs of points $(p,q)$ such that $p$ and $q$ see each other in $S$.
\end{problem}

The following definitions will also be useful for us.
\begin{definition}
We call the set of minimal points of $S$ the \emph{\NE-minimal-set} of $S$ (for ``North-East-minimal set''). The  \emph{\NW-minimal-set},  \emph{\SE-minimal-set} and  \emph{\SW-minimal-set} of $S$ are defined symmetrically (see Figure \ref{fig:minimal-sets}).  In particular, the \SE-minimal set of $S$ is the set of maximal points in $S$.
\end{definition}

\begin{figure}
    \centering
        \includegraphics{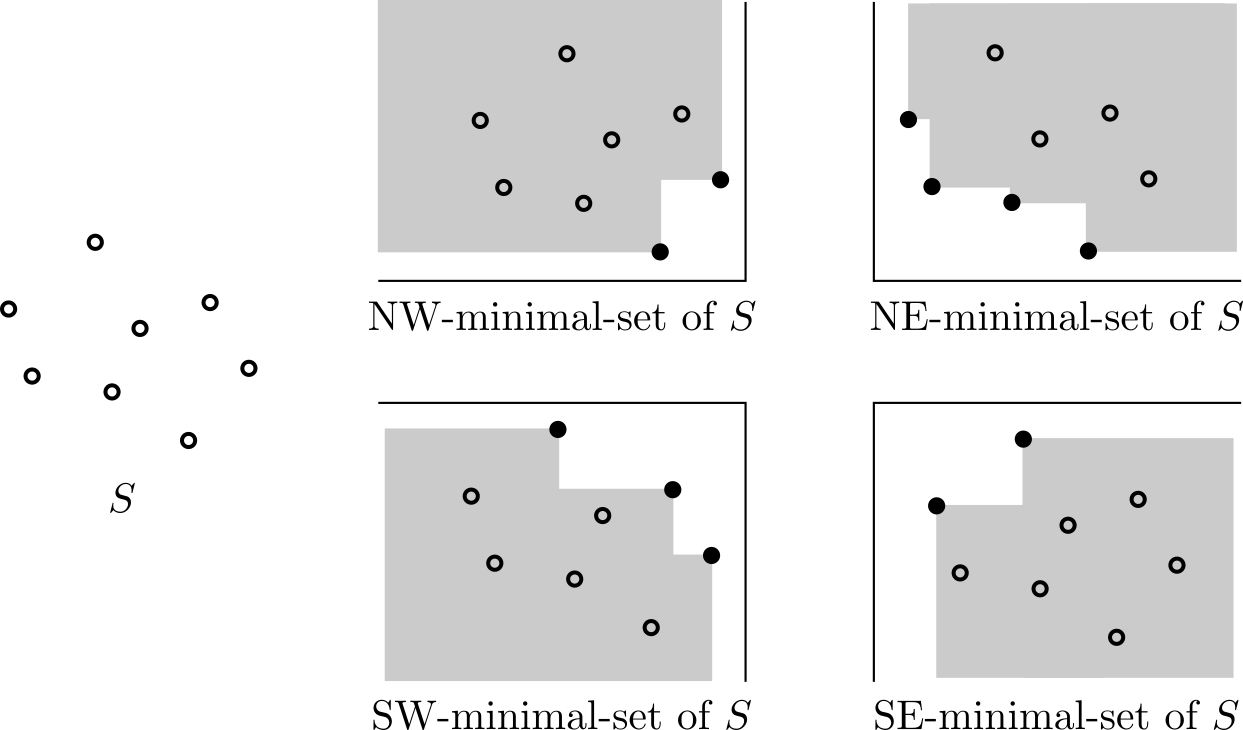}
        \caption{A set of points $S$ and the four minimal-sets of $S$. No point in the shaded regions can appear on the corresponding minimal-sets.}
        \label{fig:minimal-sets}
\end{figure}

\begin{figure}
    \centering
        \includegraphics{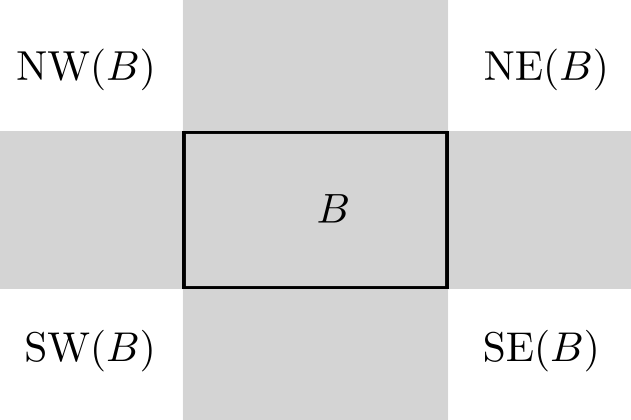}
        \caption{The four quadrants of an axis-aligned box $B$. The shaded region corresponds to $cross(B)$.}
        \label{fig:quadrants}
\end{figure}

\begin{definition}
Let $B$ be an axis-aligned box. We denote the $x$ coordinate of the right boundary (resp.\ left boundary) of $B$ as $x_{\max}(B)$ (resp.\ $x_{\min}(B)$). We denote the $y$ coordinate of the top boundary (resp.\ bottom boundary) of $B$ as $y_{\max}(B)$ (resp.\ $y_{\min}(B)$).

The \emph{cross} of $B$, denoted as $cross(B)$ is the set of points $p$ in the plane such that $x_{\min}(B) \leq x(p) \leq x_{\max}(B)$ or $y_{\min}(B) \leq y(p) \leq y_{\max}(B)$.

The \emph{quadrants} of $B$ are the connected components of $\R^2\setminus cross(B)$. We call the four components the \emph{\NE-quadrant}, \emph{\NW-quadrant}, \emph{\SE-quadrant} and \emph{\SW-quadrant}, denoted as $\NE(B)$, $\NW(B)$, $\SE(B)$ and $\SW(B)$ respectively (see Figure \ref{fig:quadrants}).
\end{definition}

All boxes we consider are axis-aligned boxes in the plane. For ease of exposition, we assume that all boxes we consider have no point of $S$ on the boundaries of their four quadrants.

\paragraph*{Lower Bound}
\begin{figure}
    \centering
    \includegraphics[scale=0.6]{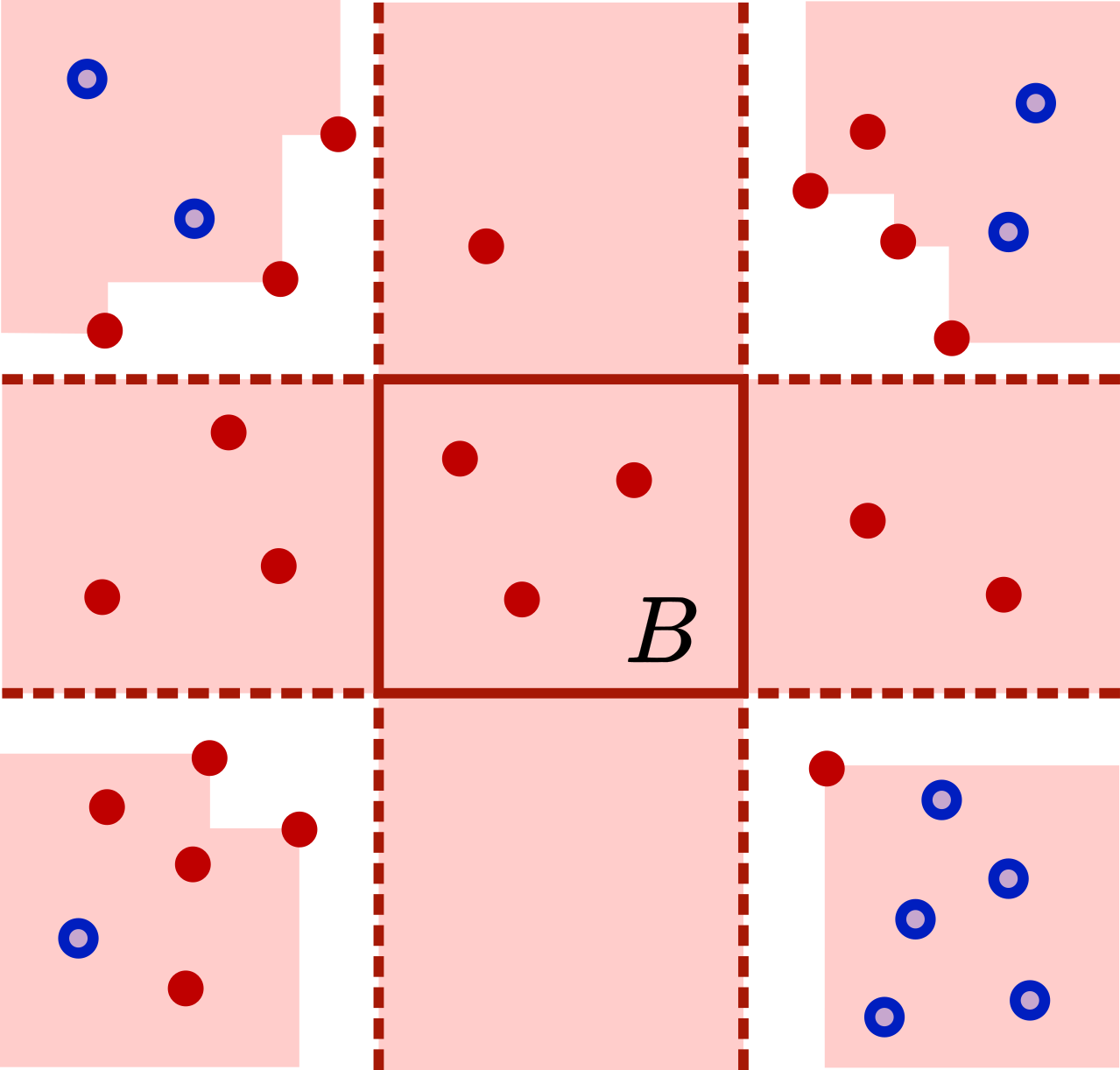}
    \caption{A red-safe box $B$.}
    \label{fig:safe-rectangle}
\end{figure}

We prove an entropy-like lower bound on the number of comparisons which have to be done to solve the problems of reporting participating points in the order-oblivious setting. The proof and terminology are largely inspired from the lower bound proofs in \cite{Afshani-2017} but some additional arguments, which we underline later, are required.
We need a few definitions in order to state our lower bound.

\begin{definition}
An axis-aligned box $B$ is \emph{red-cross-safe} (resp.\ \emph{blue-cross-safe}) for $S$ if all points in $S\cap cross(B)$ are red (resp.\ blue). It is \emph{cross-safe} if it is red-cross-safe or blue-cross-safe.

It is \emph{red-safe} if it is red-cross-safe and the \NE-minimal (resp.\ \NW-minimal, \SE-minimal, \SW-minimal) set of $S\cap\NE(B)$ (resp.\ $S\cap\NW(B)$, $S\cap\SE(B)$, $S\cap\SW(B)$) is red (see Figure \ref{fig:safe-rectangle}).
We define \emph{blue-safe} boxes similarly.
A box is \emph{safe} if it is red-safe or blue-safe.

A subset $S'\in S$ has one of these properties if it can be enclosed by a box with the property.
\end{definition}

Notice that if a subset $S'\in S$ is safe, then no point in in $S'$ participates in $S$. Thus, in an intuitive sense, a partition of non-participating points into safe subsets can be seen as a certificate for the fact that these points do not participate. The minimal entropy of such a partition is then the minimal amount of information required to encode such a certificate.

\begin{definition}
A partition $\Pi$ of $S$ is \emph{respectful} if every member $S_k \in \Pi$ is a singleton or a safe subset of $S$. The \emph{entropy} of a partition $\Pi$ is $\entropy(\Pi):= \sum_{S_k\in\Pi}\frac{|S_k|}{n}\log \frac{n}{|S_k|}$.
The \emph{structural entropy} $\entropy(S)$ of $S$ is the minimum of $\entropy(\Pi)$ over all respectful partitions $\Pi$ of $S$.
\end{definition}

We can now state our lower bound:
\begin{theorem}\label{thm:lower_bound}
For the problem of reporting participating points in the order-oblivious comparison-based model, $\OPT(S) \in \Omega(n(\entropy(S)+1))$. Consequently, for reporting red-blue pairs of visible points, $ \OPT(S) \in \Omega(n(\entropy(S)+1)+h)$, where $h$ is the size of the output.
\end{theorem}
The proof is similar to what can be found in \cite{Afshani-2017}. In their case however, if a point can be moved anywhere inside a non-safe box then it can be moved in a way that changes the expected output of the algorithm. In our case, we can do something similar but sometimes need to move multiple points to affect the expected output (but a constant number are enough). We use a simple argument about the maximum number of chips which can be placed on a tree in a constrained way to show that this has no impact on the lower bound. 

To prove Theorem \ref{thm:lower_bound}, we will need a few lemmas.
\begin{lemma}\label{lemma:safe}
Let $p$ be a point in $S$ and let $B$ be some axis-aligned box bounding $p$. If $B$ is safe, then $p$ does not interact in $S$.
\end{lemma}
\begin{proof}
Suppose without loss of generality that $p$ is red. Then $B$ is necessarily red-safe (as it cannot be blue-safe). The point $p$ could only possibly interact with a blue point $q$ lying in one of the quadrants. Because $q$ is not in the corresponding minimal-set $S'$ by definition of a red-safe box, the box spanned by $p$ and $q$ contains a point of $S'$. Thus $p$ and $q$ do not see each other in $S$.
\end{proof}

\begin{lemma}\label{lemma:force}
Let $p$ be a point in $S$ and let $B$ be some axis-aligned box. If $B$ is not safe, then we can move $p$ inside $B$ to a position where it participates in $S$.
\end{lemma}
\begin{proof}
Suppose without loss of generality that $p$ is red. The box $B$ is not safe, so in particular it is not red-safe. So either there is a blue point $q$ in $cross(B)$ or there is a blue point $q$ on the corresponding minimal-set of one of the quadrants of $B$. In the first case, we can move $p$ to a position where it almost shares the same $x$ or $y$ coordinate as $q$ such $p$ and $q$ see each other (recall that no two points share the same $x$ or $y$ coordinate). In the second case, by moving $p$ right next to the corner of $B$ corresponding to the quadrant where $q$ lies, $p$ and $q$ will see each other (recall that we assume no point of $S$ lies on the boundary of $B$ or its quadrants).
\end{proof}

\begin{lemma}\label{lemma:shield}
Let $p_1,\ldots, p_5$ be points in $S$ of the same color and let $B$ be any axis-aligned box. Then we can move these $5$ points inside $B$ such that $p_1$ does not participate in $S$.
\end{lemma}
\begin{proof}
We simply move $p_2$ close to $p_1$ such that it shields $p_1$ from seeing any point of the opposite color lying in its upper right quadrant. We shield off the other three corners similarly with $p_3,p_4$ and $p_5$.
\end{proof}

\begin{lemma}\label{lemma:chips}
Let $T$ be a rooted tree with $n$ nodes and let $c\in\N$ be an integer. Consider a set of ``chips'' each associated to a single leaf in $T$ (a leaf however can be associated to multiple chips).
Place the chips on the nodes of the tree such that every chip lies on the path from the root to its associated leaf and every root to leaf path has at most $c$ chips on it.
(Note that multiple chips can be placed on the same node.)
Then the sum $\Sigma$ of the distance from chip to associated leaf over all chips is at most $c(n-1)$.
\end{lemma}
\begin{proof}
We show the result by induction. For $n=1$ the result is trivially true. Suppose it is true for all $1\leq n'\leq n$, where $n\geq 1$ and consider a rooted tree $T$ with $n+1$ nodes. We can choose to put $c_r \leq c$ chips on the root. Each contributes at most $n-1$ to $\Sigma$, as any leaf is at distance at most $n-1$ from the root. The remaining chips have to go in the $k$ subtrees linked to the root, with respective numbers of nodes $n_1,n_2\ldots n_k$. Notice that since we already have $c_r$ chips on the root, each root to leaf path in the $k$ subtrees can only have up to $c-c_r$ chips.  By our inductive hypothesis, the contribution of the different subtrees to $\Sigma$ is at most
$(c-c_r)(n_1-1) + (c-c_r)(n_2-1) + \ldots (c-c_r)(n_k-1) \leq (c-c_r)(n_1 +n_2+\ldots n_k) = (c-c_r)(n-1)$. Together with the contribution from the chips on the root we get $\Sigma \leq (c-c_r)(n-1) + c_r(n-1) = c(n-1)$. By induction we conclude the proof.
\end{proof}

We can now prove Theorem \ref{thm:lower_bound} with a proof similar to what can be found in \cite{Afshani-2017}. In their case however, if a point can be moved anywhere inside a non-safe box then the point can be moved in a way that changes the expected output of the algorithm. In our case, we can do something similar by invoking Lemma \ref{lemma:force} or \ref{lemma:shield}, but this latter lemma requires use to be able to move multiple points of the same color (but a constant number are enough). We can use Lemma \ref{lemma:chips} to show that this does not impact the lower bound.

\begin{proof}[Proof of Theorem \ref{thm:lower_bound}]
Consider a correct algorithm $A$ in the comparison-based model. We will construct a respectful partition $\Pi$ of $S$ and ``bad'' permutation $\sigma(S)$ such that $A$ is forced to do $\Omega(n(\entropy(\Pi)+1))$ comparisons to solve the problem on $\sigma(S)$. 

We will do this by having an adversary answer the queries issued by $A$ instead of letting $A$ query the input directly. The adversary will answer queries in a way that forces $A$ to do many subsequent queries and is consistent with some permutation $\sigma(S)$. Notice that as long as this latter condition is satisfied, it makes no differences from the point of view of $A$ if it is issuing queries to the adversary or querying $\sigma(S)$ directly.

Let $S_\text{red}$ be the set of red points in $S$ and $S_\text{blue}$ be the set of blue points. First, we define two recursive partitionings on the plane similar to kd-trees \cite{Bentley-1975}, one for the red points and one for the blue points. We describe the structure for the red points. This structure will be represented by a tree $\mathcal{T}_\text{red}$ of axis-aligned boxes, constructed top to bottom. The root of $\mathcal{T}_\text{red}$ corresponds to the entire plane. If $S_\text{red}$ is not safe, we split the points around a vertical line $L$ such that the two open halfplanes defined by $L$ partition $S_\text{red}$ into two sets $S_1$ and $S_2$ of size at most $\ceil{|S_\text{red}|/2}$. The children of the root will then correspond to the two boxes thus defined. For every newly created node we repeat the procedure, partitioning the set of points by median $x$ coordinates at even levels of the tree and median $y$ coordinates at odd levels, until the points contained are a safe subset of $S$ or a singleton, at which point we define a leaf of $\mathcal{T}_\text{red}$. We define $\mathcal{T}_\text{blue}$ similarly for the blue points. The partition $\Pi_\text{kd-tree}$ of $S$ we consider is the one formed by the leaves of $\mathcal{T}_\text{red}$ and $\mathcal{T}_\text{blue}$, which is respectful by definition of the leaves.

The algorithm $A$ can issue queries to compare the $x$ or $y$ coordinate of the points at indices $i$ and $j$. It knows which indices correspond to blue points and which correspond to red. The adversary will map these indices to boxes of $\mathcal{T}_\text{red}$ or $\mathcal{T}_\text{blue}$, depending on their color. Additionally, if an index is mapped to a leaf box $B$, it will also be assigned to a point in $S\cap B$ of the corresponding color. After each query the adversary will update the mapping so that the query can be resolved by knowing only this mapping, in a way that is consistent with all previously answered queries. We let $n(B)$ denote the number of indices mapped to $B$ or a descendant of $B$. Throughout the whole execution the adversary maintains the invariant $\mathcal{I}$ that the number of indices mapped to each point in $S$ is at most $1$ and for each box $B$ in $\mathcal{T}_\text{red}$ (resp.\ $\mathcal{T}_\text{blue}$), $n(B)$ is at most $|S_\text{red}\cap B|$ (resp.\ $|S_\text{blue}\cap B|$).

A box $B \in \mathcal{T}_\text{red}$ (resp.\ $B \in \mathcal{T}_\text{blue}$) for which $n(B)=|S_\text{red}\cap B|$ (resp.\ $n(B)=|S_\text{red}\cap B|$) is said to be full. We denote the box an index $i$ is mapped to as $B_i$. An index $i$ mapped to an interior box is said to be floating. An index $i$ mapped to a leaf box is said to be fixed, and we denote the point it is assigned to as $p_i$. A point in $S$ to which no index is assigned is said to be free. 

The adversary starts by mapping all red indices to the root of $\mathcal{T}_\text{red}$ and all blue indices to the root of $\mathcal{T}_\text{blue}$. Whenever an index $i$ is mapped to a leaf box $B$, it is immediately assigned to a free point of the corresponding color in $S\cap B$. This is always possible because of the invariant $\mathcal{I}$ we maintain. When issued a query to compare, say, the $x$ coordinates of the points at indices $i$ and $j$, the adversary proceeds as follows:
\begin{enumerate}
    \item If both indices $i$ and $j$ are fixed, then the adversary simply compares $p_i$ and $p_j$ to answer the query.
    \item If either of the indices $i$ and $j$ is mapped to a non-leaf box $B$ which is at an odd depth in  $\mathcal{T}_\text{red}$ or $\mathcal{T}_\text{blue}$, map it to one of the non-full children of $B$ and move on to the following steps.
    \item Suppose that the indices $i$ and $j$ are both floating. Suppose further that the vertical line separating the children of $B_i$ lies to the left of the vertical line dividing the children of $B_j$ (otherwise, simply swap the roles of $i$ and $j$). If the left child $B^\ell_i$ of $B_i$ and the right child $B^r_j$ of $B_j$ are not full, the adversary maps $i$ to $B^\ell_i$ and $j$ to $B^r_j$. The query can then be resolved and the answer returned to $A$, as the $x$-coordinate of any point in $B^\ell_i$ is smaller than the $x$ coordinate of any point in $B^r_j$. If either $B^\ell_i$ or $B^r_j$ is full, go to the next step.
    \item If, say, $B^\ell_i$ is full, then the adversary maps $i$ to the right child $B^r_i$ of $B_i$. The comparison can not necessarily be resolved at this point, so we go back to step 2.
\end{enumerate}
If one of the indices $i$ or $j$ is fixed and the other is floating, then we can easily adapt steps $2$, $3$ and $4$, by only ever updating the mapping of the floating index.

Notice that in all these steps, the invariant $\mathcal{I}$ is maintained, as it is impossible for the two children of $B_i$ (or $B_j$) to be full. Otherwise, if $i$ is for example a red index, we would have $n(B_i) \geq n(B^\ell_i) + n(B^r_i) + 1 \geq |S_\text{red}\cap B^\ell_i| + |S_\text{red}\cap B^r_i| +1 = |S_\text{red}\cap B_i| +1$ (where the $+1$ term comes from the fact that at least one index is assigned to $B_i$, namely $i$) and the invariant was violated to begin with.

We say that an increment occurs whenever the box to which an index is mapped changes. Notice that every increment in the above procedure increase the depth of $B_i$ by exactly one for some index $i$. Thus, by the end of execution, the total number of increments is equal to the sum of the depths of the $B_i$'s over all indices $i$. Call this sum $D$.

Let us show that if $T$ is the number of queries done by $A$, then $T\in\Omega(D)$. We call the increments done in step $3$ ordinary increments. We call increments done in step $4$ exceptional increments. Every query issued by $A$ triggers $O(1)$ ordinary increments. To bound the number of exceptional increments, notice that for every box $B$ (say in  $\mathcal{T}_\text{red}$), at least $\floor{|S_\text{red}\cap B|/2}$ ordinary increments have to occur at $B$ before any exceptional increment can occur, and that at most $\ceil{|S_\text{red}\cap B|/2}$ exceptional increments can occur at $B$ after that point. So the total number of exceptional increments can be bounded by the total number of ordinary increments, which is $O(T)$. The increments done in step $2$ are either the direct results of a query or were preceded by some unique exceptional query to which they can be charged, so their number is also $O(T)$. Thus, the total number of increments $D$ is in $O(T)$, or in other words, $T\in\Omega(D)$.

Next, notice that after the execution of $A$, if there is any index which is mapped to an interior box $B$, we can arbitrarily remap it to one of the children of $B$, then repeat the procedure until all indices are fixed. The resulting mapping from indices to points defines a permutation $\sigma(S)$ consistent with all answers given to the queries by the adversary and which maps red indices with red points, blue indices with blue points. Thus from the point of view of $A$, we can consider it has queried $\sigma(S)$ directly. In other words, the adversary did not ``cheat''. Let $D_\text{post}$ be the sum of the depth of $B'_i$ over all indices $i$ after this post-processing step. Every index is associated to a leaf and every leaf is full. For any box $B$ in, say, $\mathcal{T}_\text{red}$ at depth $d$ we have $|B\cap S_\text{red}| \geq \floor{\frac{n}{2^d}}$, thus $d \in \Omega(\log \frac{n}{|B\cap S_\text{red}|})$. It follows that
\begin{align*}
    D_\text{post} &\in\Omega\left( \mspace{-30mu} \sum_{\mspace{38mu} B \text{ leaf in }\mathcal{T}_\text{red}} \mspace{-40mu} |B\cap S_\text{red}| \log \frac{n}{|B\cap S_\text{red}|} + 
     \mspace{-30mu} \sum_{\mspace{38mu} B \text{ leaf in }\mathcal{T}_\text{blue}}  \mspace{-40mu} |B\cap S_\text{blue}| \log \frac{n}{|B\cap S_\text{blue}|}\right) \\
    &= \Omega \left( \mspace{-30mu}\sum_{\mspace{38mu} S_k \in \Pi_\text{kd-tree}} \mspace{-30mu} |S_k| \log \frac{n}{|S_k|}\right)   \\
    &= \Omega(n\entropy(\Pi_\text{kd-tree})).
\end{align*}

It remains to prove a lower bound on $D$ using the lower bound we have on $D_\text{post}$. We show that before the post-processing described in the previous paragraph, for every root to leaf path in $\mathcal{T}_\text{red}$ or $\mathcal{T}_\text{blue}$, the number of indices mapped to an interior box on the path is at most $4$. Indeed, suppose there was such a path in, say, $\mathcal{T}_\text{red}$ with $5$ or more indices mapped to interior boxes. Call these indices $i_1,i_2,\ldots,i_5$, and let $p_1,\ldots,p_5$ be the points they are associated with after post-processing. Notice that we can freely move the point $p_1$ inside $B_{i_1}$ for example, without changing the answer to any of the queries done by $A$. Let $B$ be the box of maximum depth among $B_{i_1},\ldots, B_{i_5}$. Because all other boxes are ancestors of $B$ in $\mathcal{T}_\text{red}$, they all contain $B$. Thus, we can move all five points inside $B$ without changing the answer to any of the queries done by $A$. Because $B$ is not safe we can apply either Lemma \ref{lemma:force} or Lemma \ref{lemma:shield} to move the points inside $B$ in a way that creates a new point set $S'$ with a different expected output. Then $A$ will not compute the correct solution on $S'$, which contradicts its definition as a correct algorithm. 

Now, for each index $i$ mapped to an interior node, place a ``chip'' on the corresponding tree, at the node to which $i$ maps before post-processing. Associate the chip with the leaf to which $i$ gets mapped after post-processing. We know that every chip lies on the path from the root to its associated leaf and that every root to leaf path in $\mathcal{T}_\text{red}$ or $\mathcal{T}_\text{blue}$ has at most $4$ chips on it. The sum of the distance from chip to associated leaf over all chips is $D_\text{post}-D$. By Lemma \ref{lemma:chips}, we get $D_\text{post}-D \leq 4(n-1)$, that is, $D\geq D_\text{post}-4(n-1)$. Thus,
\begin{align*}
    T &\in \Omega(D)\\
    &\subset \Omega(D_\text{post}-4n)\\
    &\subset \Omega(n\entropy(\Pi_\text{kd-tree})-4n)\\
    &\subset \Omega(n(\entropy(\Pi_\text{kd-tree})-4)).
\end{align*}

Combined with the trivial $\Omega(n)$ bound we get a $\Omega(n(\entropy(\Pi_\text{kd-tree})+1))$ bound, which in turn implies a $\Omega(n(\entropy(S)+1))$ lower bound, as $\Pi_\text{kd-tree}$ is a respectful partition of $S$.
\end{proof}

Notice that this lower bound implies the following:
\begin{corollary}\label{cor:simple_lower_bound}
For the problem of reporting participating points in the order-oblivious comparison-based model, $\OPT(S) \in \Omega(h\log n)$, where $h$ is the number of points which participate in $S$.
\end{corollary}
\begin{proof}
Consider a respectful partition $\Pi$ of $S$. By definition of entropy we have $\entropy(\Pi) = \sum_{S_k\in \Pi} \frac{|S_k|}{n} \log \frac{n}{|S_k|}$. By definition of a respectful partition and by Lemma \ref{lemma:safe}, every point which participates in $S$ will be in a singleton $S_k$. By summing only over these $h$ points, we get $n\entropy(\Pi) \geq n\sum_{i=1}^h \frac{1}{n} \log n = h\log n$. As this is true for any respectful partition, we have $n\entropy(S) \geq h\log n$.
\end{proof}

This means in particular that if we can report all points which participate in $S$ in $O(n(\entropy(S)+1))$ time, then we can afford to run an algorithm in worst case optimal $O(h\log h + j)$ time on these reported points to compute the $j$ red-blue pairs which see each other. The resulting algorithm will be instance-optimal in the order-oblivious model. 

\section{Instance optimal comparison-based algorithm in the order-oblivious setting}\label{sec:upper_bound}
In this section we present a comparison-based algorithm for reporting participating points with a runtime matching the lower bound we saw in the previous section (note that worst-case optimal $\bigO(n\log n + h)$ algorithms are easy to obtain by the method of \cite{Guting-1989}). Once again, the main algorithm will be very similar to what is done in \cite{Afshani-2017}, however their results do not directly apply here. The main difficulty in adapting their algorithm to our case is that the relation we consider here is not decomposable. More precisely, if we know that some point $p$ does or does not participate in $S'$ and $S''$, we cannot use this to decide if $p$ participates in $S'\cup S''$.
The bulk of the work here will thus be to preprocess the set $S$ in order to make the safety tests decomposable, while keeping our preprocessing time within $O(n(\entropy(S)+1))$.

\subsection{The main algorithm}

Before we go into detail about how to preprocess the points let us see how, if we can build the right data structure, we can use it to report the participating points in $S$ in $O(n(\entropy(S)+1))$ time.
We will need the following definition and easy observations:
\begin{definition}
Let $S$ be a set of red and blue points. A subset $S'\subset S$ \emph{conforms with $S$} if it contains all points which participate in $S$.
\end{definition}
\begin{observation}\label{obs:conform}
Let $S$ be a set of of red and blue points and let $S'$ be a subset which conforms with $S$. Then an axis-aligned box is safe for $S$ if and only if it is safe for $S'$. Moreover, a point participates in $S$ if and only if it participates in $S'$.
\end{observation}

\begin{observation}\label{obs:subbox}
If $B$ is a safe box for $S$, then any sub-box of $B$ is safe for $S$.
\end{observation}

We have the following algorithm and theorem, adapted from \cite{Afshani-2017}, where $\delta$ is a constant to be chosen later:

\begin{algorithm}[H]
\DontPrintSemicolon
    \KwInput{A point set $S$ of size $n$}
    Set $Q = S$. \;
    \For{$j=0,1,\ldots \floor{\log(\delta \log n)}$}{
        Partition the points in $Q$ using a kd-tree to get $r_j=2^{2^j}$ subsets $Q_1,\ldots, Q_{r_j}$ of size at most $\ceil{|Q|/r_j}$, along with corresponding bounding boxes $B_1,\ldots, B_{r_j}$.\;
        \For{$i=0,1,\ldots, r_j$}{
            \If{$B_i$ is safe for $Q$}{
                Prune all points in $Q_i$ from $Q$. \;
            }
        }
    }
    Solve the reporting problem for the remaining set $Q$ directly in $O(|Q|\log|Q|)$ time. \;
    
\caption{Reporting participating points}\label{alg:main}
\end{algorithm}

\begin{theorem}\label{thm:main_algo}
Let $S$ be a set of $n$ points in general position.
Suppose we have preprocessed $S$ such that for any subset $S'\subset S$ containing all points which participate in $S$ we can test if an axis-aligned box is safe for $S'$ in $O(n^{1-\alpha})$ time, plus the cost of a constant number of range-emptiness queries on $S'$.\\
Then Algorithm \ref{alg:main} can report all points which participate in $S$ in $O(n(\entropy(S)+1))$ time.
\end{theorem}
We reiterate the proof for the sake of completeness and to underline how Observation \ref{obs:conform} and our additional assumptions on preprocessing factor into it:
\begin{proof}
By Observation \ref{obs:conform} and Lemma \ref{lemma:safe}, we only ever prune points which do not participate in the original set $S$ and we never modify which points participate among those that remain. Thus the algorithm invoked at line $6$ will compute the correct output. 

By assumption, testing a box for safety in $Q$ can be done in $O(n^{1-\alpha})$ time plus the cost of a constant number of range-emptiness queries on $Q$. Using a simple and ingenious trick by T.\ Chan  \cite{Chan-1996}, we can perform $r$ orthogonal range emptiness queries on a set of size $m$ in $O(m\log r + r^{O(1)})$ time. Thus, the $r_j$ tests of lines $3$ and $4$ can be done in $O(|Q|\log r_j + r_j^{O(1)} + r_j n^{1-\alpha})$ time. As $r_j < n^\delta$, by taking $\delta$ small enough the $r_j^{O(1)} + r_j n^{1-\alpha}$ term can be made sublinear. As the outer loop of the algorithm is only executed $O(\log(\log n))$ times the total contribution of these terms over the whole algorithm can also be made sublinear and thus negligible.
Line 2 can be done in $O(|Q|\log r_j)$ time by the classical recursive algorithm to compute kd-trees.

Now let $n_j$ be the number of points in $Q$ just after iteration $j$. The runtime of the algorithm is in $O(\sum_{j} n_j\log r_{j+1})$. This includes the final step at line $6$, as for $j=\floor{\log(\delta \log n)}$ (i.e.\ after the last iteration of the outer loop) we have 
$O(|Q|\log|Q|) \subset O(n_j\log n^{2\delta}) = O(n_j\log r_{j+1})$.
Let $\Pi$ be a respectful partition of $S$ and consider $S_k \in \Pi$. At iteration $j$ all subsets $Q_i$ lying entirely inside the bounding box of $S_k$ are pruned by Observation \ref{obs:subbox}. Since the bounding box of $S_k$ intersects at most $O(\sqrt{r_j})$ cells of the kd-tree, the number of points in $S_k$ remaining after iteration $j$ is $min\{|S_k|, O(\sqrt{r_j} \cdot n/r_j)\}$ = $min\{|S_k|, O(n/\sqrt{r_j})\}$.
The $S_k$'s cover the entire point set so by double summation we have:

\begin{align*}
\sum_j n_j \log r_{j+1} &\leq \sum_j \sum_k \min\left\{|S_k|, O\left(n/\sqrt{2^{2^j}}\right)\right\}\cdot 2^{j+1}\\
&= \sum_k \sum_j \min\left\{|S_k|, O\left(n/2^{2^{j-1}}\right)\right\}\cdot 2^{j+1}\\
&\in O\left(\sum_k \mspace{-50mu}\sum_{\mspace{70mu}j \leq \log (2\log (n/|S_k|))}\mspace{-70mu} |S_k|\cdot 2^{j} 
\mspace{10mu}+\mspace{-60mu}
\sum_{\mspace{70mu}j > \log (2\log (n/|S_k|))}\mspace{-70mu}n\cdot2^{j}/2^{2^{j-1}} \right)\\
&\in O\left(\sum_k |S_k|(\log(n/|S_k|)+1)\right) = O(n(\entropy(S)+1)). \qedhere
\end{align*}%
\end{proof}

\subsection{Cross-safety tree}

\begin{figure}
    \centering
    \includegraphics{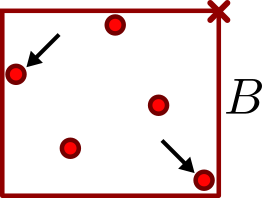}
    \caption{A cell of red points and the corresponding box $B$. The relevant points are indicated by arrows. The box point is indicated by a cross.}
    \label{fig:box_point}
\end{figure}

In order to solve our problem instance-optimally (in the order-oblivious comparison-based model), we want to design a data-structure which allows us to quickly test if a given axis aligned box $B$ is safe for $S$.
To make the presentation clearer, we focus on testing if $B$ is red-cross-safe and the $\NE$-minimal-set of $\NE(B)\cap S$ is red. This can then be repeated symmetrically for $\NW(B),\SE(B)$ and $\SW(B)$ to test if $B$ is red-safe (and similarly for testing if $B$ is blue-safe).  
We will see how to build and query the following data structure:

\begin{definition}
A \emph{cross-safety tree $T_S$ on $S$} is a recursive partitioning on the plane similar to a kd-tree where we stop subdividing the points once we have reached a cross-safe subset of points. The root of $\mathcal{T}_\text{red}$ corresponds to $S$. If $S$ is not cross-safe, we split the points around a vertical line $L$ such that the two open halfplanes defined by $L$ partition $S$ into two sets $S_1$ and $S_2$ of size at most $\ceil{|S|/2}$. The children of the root will then correspond to $S_1$ and $S_2$. For every newly created node we repeat the procedure, partitioning the set of points by median $x$ coordinates at even levels of the tree and median $y$ coordinates at odd levels, until the points contained are a cross-safe subset of $S$.

The \emph{cell} of a point $p$, denoted as $C_p$ is the the subset of points contained in the same leaf as $p$. A cell of $T_S$ is red (resp.\ blue) if the points it contains are red (resp.\ blue).

The \emph{box of a point $p$}, denoted as $B_p$ is the smallest axis-aligned box containing all points in $C_p$ (slightly extended to enforce our assumption of only considering boxes for which there are no points on the boundary of their four quadrants). A box of $T_S$ is red (resp.\ blue) if the points it contains are red (resp.\ blue).

The \emph{box-point of $p$} is the top-right point of the box of $p$, and has the same color as $p$.

A point is \emph{relevant} if it has the minimum $x$ or $y$ coordinate among all points in its box (or equivalently, in its cell).

(See \cref{fig:box_point} for an illustration of these definitions.)

Each node $u$ in the tree stores:
\begin{itemize}
    \item The set of point it contains, which we denote as $P_u$.
    \item The smallest axis-aligned bounding box of $P_u$, which we denote as $B_u$.
    \item The red box-points of minimum $x$ and $y$ coordinates among all box-points associated with a red point in $P_u$.
    \item The subset of all relevant blue points in the minimal set of $P_u$, sorted by $x$-coordinate.
\end{itemize}

\end{definition}

Note that if the points in a minimal set are sorted by $x$ coordinate then they are also sorted (in reverse order) by $y$ coordinate.

\subsection{Querying a cross-safety tree}

Before we see how to build a cross-safety tree on $S$ efficiently, let us see how we make queries on a node $u$ of $T_S$. A query consists of a lower range $R_L = [x_L,+\infty]\times [y_L,+\infty]$ and an upper range $R_U = [-\infty,x_U]\times [-\infty,y_U]$ such that $R_L\cap R_U \neq \emptyset$ and neither the boundaries of $R_L$ nor $R_U$ intersect any blue box of $T_S$. It returns:
\begin{itemize}
    \item $rx_u$, the minimum $x$-coordinate of any red box-points associated with a red point in $P_u \cap R_L$ (set to $+\infty$ if there are no red points in $P_u \cap R_L$).
    \item $ry_u$, the minimum $y$-coordinate of any red box-points associated with a red point in $P_u \cap R_L$ (set to $+\infty$ if there are no red points in $P_u \cap R_L$).
    \item $bx_u$, the minimum $x$-coordinate of any blue points in the minimal set of $P_u \cap R_L \cap R_U$ (set to $+\infty$ if there are no blue points in the minimal set of $P_u \cap R_L \cap R_U$).
    \item $by_u$, the minimum $y$-coordinate of any blue points in the minimal set of $P_u \cap R_L \cap R_U$ (set to $+\infty$ if there are no blue points in the minimal set of $P_u \cap R_L \cap R_U$).
\end{itemize}

Observe the following:
\begin{observation}
\label{obs:line-through-point}
If a horizontal or vertical line passes through a red point then it does not intersect any blue-cross-safe box. The same applies when ``red'' and ``blue'' are swapped.
\end{observation}

\begin{observation}\label{obs:domination-in-safe-box}
Let $B$ be a red-cross-safe box. Then all points in $B$ dominate (resp. are dominated by) the same subset of blue points in $S$. The same applies when ``red'' and ``blue'' are swapped.
\end{observation}

We will need a few additional lemmas.
\begin{lemma}\label{lemma:relevant-points-suffice}
The points corresponding to $bx_u$ and $by_u$ are relevant points of $T_S$.
\end{lemma}
\begin{proof}
Let $P = P_u \cap R_L \cap R_U$, and suppose there is a blue point on the minimal set of $P$.
Let $p$ be the leftmost point in $P$ which does not dominate any red point in $P$. Suppose that $bx_u$ is not equal to $x(p)$. The only way this can happen is if $p$ is not on the minimal set of $P$, meaning that there is a blue point $q$ which is dominated by $p$ (as $p$ does not dominate any red point). In particular, $q$ lies to the left of $p$. By definition of $p$, $q$ thus dominates a red point. But if $q$ dominates a red point and $p$ dominates $q$, then $p$ dominates a red point, which contradicts the definition of $p$. Thus $bx_u=x(p)$.
Moreover, if $p$ does not dominate any red point, then no point in its box dominates a red point, as the box is blue-cross-safe. Because the boundaries of $R_L$ and $R_U$ do not intersect any blue box, the whole box of $p$ is contained in $P$. Thus, by definition of $p$, it is the leftmost point in its box and it is relevant.
The same reasoning shows that $by_u$ is the $y$ coordinate of a relevant point.
\end{proof}

\begin{lemma}\label{lemma:case-inside}
If the bounding box $B_u$ of points in $P_u$ lies entirely in $R_L$, then we can return the necessary information in $O(\log n)$ time.
\end{lemma}
\begin{proof}
In this case, $rx_u$ and $ry_u$ are already stored in the node, so we can return them in constant time. 
By Lemma \ref{lemma:relevant-points-suffice}, the point giving $bx_u$ is the relevant blue point with minimum $x$ coordinate among all blue points in the minimal set of $P_u \cap R_U$. This is also the relevant blue point with maximum $y$ coordinate among all blue points in the minimal set of $P_u \cap R_U$. 

We can do a binary search through the relevant blue points in the minimal set of $P_u$ to find the point $p$ below the line $y=y_U$ with maximum $y$ coordinate. If this point is not in $R_U$ (because $x(p) > x_U$) then no relevant blue point in the minimal set of $P_u$ is in $R_U$, as all other relevant blue points $q$ in the minimal set of $P_u$ with $y(q)\leq y_U$ will have $x(q) > x(p) > x_U$. In this case we set $bx_u$ to $+\infty$. Otherwise $bx_u = x(p)$.

We can find $by_u$ similarly with a single binary search.
\end{proof}

\begin{lemma}\label{lemma:case-interior-node}
If $u$ is not a leaf of $T_S$ and $B_u$ intersects the boundary of $R_L$, then we can return the necessary information after querying the children of $u$ with the same lower range $R_L$ (but possibly different upper ranges) and a constant amount of additional work.
\end{lemma}
\begin{proof}
Suppose without loss of generality that the children of $u$ split $P_u$ by a vertical line (the situation for a horizontal line is similar). Let $v$ be the child corresponding to the left half-plane and $w$ the one corresponding to the right half-plane.
Let us focus on computing $bx_u$, as this is the one requiring the most care.

Querying $v$ with the same $R_L$ and $R_U$ returns some values $rx_v$, $ry_v$, $bx_v$ and $by_v$. A blue point $p$ on the minimal set of $P_w\cap R_L \cap R_U$ is a blue point on the minimal set of $P_u\cap R_L \cap R_U$ if and only if $y(p) \leq by_v$ and $p$ does not dominate any red point in $P_v\cap R_U$. By Observation \ref{obs:domination-in-safe-box} this is equivalent to saying that $y(p) \leq by_v$ and $y(p) \leq ry_v$. 

Thus we can compute $bx_u$ by setting $R'_U = [-\infty,x_U]\times [-\infty,\min\{y_U, ry_v\}]$, querying $w$ with $R_L$ and $R'_U$ to get values $rx_w$, $ry_w$, $bx_w$ and $by_w$, then setting $bx_u = \min\{bx_v, bx_w\}$. Notice that by Observation \ref{obs:line-through-point}, we are allowed to query $w$ with $R'_U$ as its boundary does not intersect any blue box of $T_S$.
It is then easy to see that $rx_u = \min\{rx_v, rx_w\}$, $ry_u = \min\{ry_v, ry_w\}$ and $by_u = \min\{by_v, by_w\}$.
\end{proof} 

\begin{lemma}\label{lemma:case-leaf-side}
If $u$ is a leaf of $T_S$ and $B_u$ intersects the lower boundary or the left boundary of $R_L$ but not both simultaneously, then we can return the necessary information in constant time.
\end{lemma}
\begin{proof}
In this case, we know that $B_u$ is a red box, as by assumption the boundary of $R_L$ does not intersect any blue box of $T_S$. Thus $B_u$ contains no blue points and we know $bx_w=by_w=+\infty$.
Suppose without loss of generality that $B_u$ intersects the left boundary of $R_L$. Then the rightmost point of $B_u$ is in $R_L$, and thus $P_u\cap R_L$ is non-empty. Because all points in $P_u$ have the same box-point $p=(p_x,p_y)$, we have $rx_u = p_x$ and $ry_u = p_y$.
\end{proof}

\begin{lemma}\label{lemma:case-leaf-corner}
If $u$ is a leaf of $T_S$ and $B_u$ intersects both the lower boundary and the left boundary of $R_L$, then we can return the necessary information after a single orthogonal range-emptiness test.
\end{lemma}
\begin{proof}
Again, we know that $B_u$ is a red box and all points in $P_u$ have the same box point $p=(p_x,p_y)$. To know if we need to set $rx_u = ry_u = +\infty$ or $rx_u = p_x$ and $ry_u = p_y$, we simply need to test if there is a (necessarily red) point in $B_u\cap R_L$. This requires a single range-emptiness test.
\end{proof}

By applying the relevant result among Lemmas \ref{lemma:case-inside}, \ref{lemma:case-interior-node}, \ref{lemma:case-leaf-side} and \ref{lemma:case-leaf-corner} recursively we get:
\begin{theorem}
We can query the root of a cross-safety tree $T_S$ in $O(\sqrt{n}\log n)$ time plus the cost of a single range-emptiness test.
\end{theorem}

As a corollary, we get:

\begin{corollary}\label{cor:ne-safety-test}
Let $B$ be an axis aligned box. We can query a cross-safety tree $T_S$ to test if $B$ is red-$\NE$-safe for $S$ in $O(\sqrt{n}\log n)$ time plus the cost of a constant number of orthogonal range-emptiness tests on $S$.
\end{corollary}
\begin{proof}
With two orthogonal range-emptiness queries on the blue points of $S$ and one on the red points of $S$ we can test if $B$ contains at least one red point and $\cross(B)$ contains only red points (that is, test if $B$ is red-cross-safe for $S$). If $B$ is not red-cross-safe for $S$ we can immediately return ``No''. We assume from now on that it is.

Let $R_L$ be the range corresponding to $\NE(B)$, and let $R_U = [-\infty,+\infty]\times [-\infty,+\infty]$. Because $R$ is red-cross-safe, it is easy to see that the boundary of $R_L$ does not intersect any blue box of $T_S$ (this is also trivially true for $R_U$). 
Thus, we can query the root $u$ of $T_S$ with $R_L$ and $R_U$ to get the $x$ coordinate $bx_u$ of the blue point with minimum $x$ coordinate among all blue points on the minimal set of $\NE(B)\cap S$, or $+\infty$ if no such point exists. In particular, this allows us to test if such a point exists.
\end{proof}

We also have the following:
\begin{lemma}\label{lemma:cross-safety-tree-conform}
Let $S'\subset S$ be a subset which conforms with $S$ and let $B$ be an axis-aligned box. If $S'\cap B \neq \emptyset$, then we can replace all orthogonal range-emptiness tests on $S$ with the same tests on $S'$ in the procedure described in Corollary \ref{cor:ne-safety-test} (including the tests done while querying $T_S$) without affecting the outcome. 
\end{lemma}
\begin{proof}
If there are both red and blue points in $S\cap\cross(B)$, then at least one of these blue points participates in $S$. Because $S'$ conforms with $S$, this blue point is also in $S'$, so $S'$ is not red-cross-safe. The converse is trivially true. Thus, the initial three range-emptiness tests return the same results on $S$ and $S'$.

Now consider the query done on $T_S$. If the corner of the lower range $R_L$ does not intersect a red leaf box of $T_S$, then no range-emptiness test is performed and the claim holds. Now suppose $R_L$ intersects a red leaf box $B$ of $T_S$. Let $S''$ be the set of points in $S$ where we remove all points in $S \cap B \cap R_L$ which are not in $S' \cap B \cap R_L$. Notice that by replacing the range-emptiness query on $S$ with one on $S'$, the procedure behave exactly like querying a cross-safety tree on $S''$. Because $S'\subset S'' \subset S$ we know that $S''$ conforms with $S$ and thus by Observation \ref{obs:conform} the claim holds.
\end{proof}

Thus, this data-structure fits the prerequisites of Theorem \ref{thm:main_algo}, and we can use it to get an algorithm solving the problem in $O(n(\entropy(S)+1))$ time after having built it. The only missing ingredient to get an instance-optimal algorithm is building the data-structure within the same asymptotic runtime. We show in the following section that we can indeed do this. 

\subsection{\texorpdfstring{Construction in $O(n(\entropy(S)+1))$ time}{Construction in O(n(H(S)+1)) time}}

Rather than focusing on constructing the cross-safety tree specifically, we start with a bit more general setting.

\begin{theorem}\label{thm:C-kd-tree}
Let $S$ be a set of points. Let $C$ be a property on axis-aligned bounding boxes of the plane such that for boxes $B_2\subset B_1$, if $C(B_1)$ is true then $C(B_2)$ is true. (Note that $C$ can depend on $S$).

A partition $\Pi$ of $S$ is $C$-respectful if every set in $\Pi$ is a singleton or can be enclosed by an axis aligned bounding box $B$ such that $C(B)$ is true.

Let $\entropy_C(S)$ be the minimum of $\entropy(\Pi)$ over all $C$-respectful partitions of $S$.

If property $C(B)$ can be tested in $O(|S\cap B|)$ time when given access to $S\cap B$, then a kd-tree $T_S^C$ with stopping condition $C$ on the leaves can be built in $O(n(\entropy_C(S)+1))$ time.

This remains true if for each node in the tree we do an additional linear amount of work (in the number of points considered at that node).
\end{theorem}
The proof is similar in spirit to that of Theorem $\ref{thm:main_algo}$, although simpler.
\begin{proof}
Consider the classical top-down recursive approach to construct a kd-tree on $S$ (with linear-time median, selection), where we stop subdividing points once we have reached a bounding-box $B$ with property $C$.
Consider any $C$-respectful partition $\Pi$ of $S$. Let $S_k \in \Pi$ and let $B_k$ be a corresponding bounding box with the property $C$. At the $j$'th level of the recursion, we have partitioned the plane into $O(2^j)$ boxes each containing at most $\ceil*{n/2^j}$ points still to be considered. Any box $B$ of $T_S^C$ which is entirely contained in $B_k$ has property $C$ and can be set as a leaf. In other words the box $B$ does not need to be recursed on and the points in $B\cap S$ are not considered at level $j$ or lower. Because $B_k$ intersects at most $O(\sqrt{2^j})$ boxes of $T_S^C$ at level $j$, the number of points in $S_k$ to consider at level $j$ is $min\{|S_k|, O(\sqrt{2^j} \cdot n/2^j)\}$ = $min\{|S_k|, O(n/\sqrt{2^j})\}$.
At each level, the amount of work to be done is linear in the number of points to consider. The $S_k$'s cover the entire point set so by double summation we get that the runtime is in

\begin{align*}
O\left(\sum_j \sum_k \min\left\{|S_k|, n/\sqrt{2^j}\right\}\right)
&= O\left(\sum_k \sum_j \min\left\{|S_k|, n/\sqrt{2^j}\right\}\right)\\
&= O\left(\sum_k \mspace{-20mu}\sum_{\mspace{30mu}j \leq 2\log (n/|S_k|)}\mspace{-30mu} |S_k| 
\mspace{10mu}+\mspace{-30mu}
\sum_{\mspace{30mu}j > 2\log (n/|S_k|)}\mspace{-30mu}n/\sqrt{2^j} \right)\\
&= O\left(\sum_k |S_k|(\log(n/|S_k|)+1)\right) \\
&\subset O\left(n(\entropy_C(S)+1)\right). \qedhere
\end{align*}
\end{proof}

Note that the proof generalizes easily to any constant dimension $d>0$.
We can apply this theorem to get the following:
\begin{lemma}
A set of points $S$ can be preprocessed in $O(n(\entropy(S)+1))$ time so that for any subset $S_k\subset S$, we can test if all points in $S_k$ lie in a common vertical slab containing only points of $S$ of the same color in $O(S_k)$ time.
\end{lemma}
\begin{proof}
Start by projecting all points in $S$ on the horizontal line $y=0$ to get a new set of points $S'$. Apply the previous theorem in 1 dimension with the property $C(I)$ being ``all points of $S'$ in the interval $I$ are of the same color''. This requires $O(n(\entropy_C(S')+1))$ time. It is easy to see that with the constructed data-structure we can assign to each point the maximal monochromatic interval that contains it in linear time. Then we can test if all points in $S_k$ lie in a common vertical slab containing only points of $S$ of the same color in $O(S_k)$ time. Moreover, because all points in a safe subset of $S$ lie in such a vertical slab, we have $\entropy_C(S') \leq \entropy(S)$ and so the claim holds.
\end{proof}

Which in turn implies:
\begin{theorem}
A cross-safety tree on $S$ can be constructed in $O(n(\entropy(S)+1))$.
\end{theorem}
\begin{proof}
Call $C(B)$ the property ``$B$ is a cross-safe axis-aligned box''. Using the previous lemma (and the analogous result for horizontal slabs), we can preprocess $S$ in $O(n(\entropy(S)+1))$ such that property $C(B)$ can be tested in $O(|B\cap S|)$ when given access to $B\cap S$. Notice also that when constructing a cross-safety tree $T_S$ on $S$, all information to be stored at a node $u$ can be computed in $O(|B_u\cap S|)$ time (using the information stored in the children in case of an interior node). Thus we can apply Theorem \ref{thm:C-kd-tree} to construct $T_S$ in $O(n(\entropy_C(S)+1))$. Because every safe box is also cross-safe, we have $\entropy_C(S) \leq \entropy(S)$, which concludes the proof.
\end{proof}

Finally, putting this together with \ref{lemma:cross-safety-tree-conform} and Theorem \ref{thm:main_algo} we get the main result.
\begin{theorem}
All points participating in $S$ can be reported in $O(n(\entropy(S)+1))$ time. In other words, there is an instance-optimal algorithm for this problem in the order-oblivious comparison-based model.
\end{theorem}
One thing to note here is that while this guarantees that the algorithm is optimal with respect to any parameter of the instance which does not depend on the order of the input points, it is not immediately obvious that it runs in $O(n\log h)$ time, where $h$ is the number of points to report (we only know that if there is an algorithm in the comparison-based model running within this time bound, then so does ours). The following theorem proves this fact.

\begin{theorem}
Let $S$ be an instance of the Reporting participating points problem and let $h$ be the number of points which participate in $S$. Then $n(\entropy(S)+1) \in O(n\log h)$.
\end{theorem}
\begin{proof}
We prove this by showing that there is a respectful partition $\Pi$ of $S$ such that $n(\entropy(\Pi)+1) \in O(n\log h)$. By definition of $\entropy(S)$ the claim then holds.

We define $h+1$ vertical slabs $V_1,V_2,\ldots,V_{h+1}$ of the plane as follows.
\begin{itemize}
    \item Let $\epsilon>0$ be a constant smaller than the difference in $x$ coordinates between any two distinct points in $S$.
    \item Let $p_1,p_2,\ldots,p_h$ be the points participating in $S$, ordered by $x$ coordinate.
    \item Let $H_1$ be the set of points in the plane with $x$ coordinate at most $x(p_1)-\epsilon$. Let $H_{h+1}$ be the set of points in the plane with $x$ coordinate at least $x(p_h)+\epsilon$.
    \item For $1<i<h+1$, let $H_i$ be the set of points in the plane with $x$ coordinate between $x(p_{i-1})+\epsilon$ and $x(p_{i})-\epsilon$.
\end{itemize}
We define $h+1$ horizontal slabs $V_1,V_2,\ldots,V_{h+1}$ similarly by ordering the participating point by $y$ coordinate.

The pairwise intersections of these slabs define $(h+1)^2$ disjoint axis-aligned boxes. We let $B_1, B_2, \ldots, B_k$ be the set of such boxes which contain at least one point of $S$ (which is necessarily non-participating). It is clear that $k \leq \min\{n-h,(h+1)^2\}$. Moreover, for any $1\leq i \leq k$, the box $B_i$ is safe. Indeed, all points in $cross(B_i)$ are of the same color, otherwise at least one of these points would participate in $S$, which contradicts the definition of $B_i$. Without loss of generality, suppose this color is red. Now, suppose there was a blue point $b$ on, say, the $\NE$-minimal-set of $S\cap \NE(B_i)$. Then, because all points in $cross(B_i)$ are red and do not participate in $S$, the box spanned by $b$ and the lower-left corner of $B_i$ contains no point of $S$. This contradicts the fact that $B_i$ contains a point of $S$ by definition. Thus there is no blue point on the $\NE$-minimal-set of $S\cap \NE(B_i)$ (and similarly for the three other quadrants of $B_i$). So $B_i$ is safe. For all $1\leq i\leq k$, we define $S_i = S\cap B_i$.

The partition we consider is $\Pi = \{S_i \mid 1\leq i \leq k\} \cup \{\{p\}\mid p \text{ participates in } S\}.$ It is respectful by the previous discussion. We have
\begin{align*}
    \entropy(\Pi) &= \sum_{i=1}^k \frac{|S_i|}{n}\log \frac{n}{|S_i|} + \sum_{p \text{ participates in } S} \frac{1}{n}\log n\\
    &=  \sum_{i=1}^k \frac{|S_i|}{n}\log \frac{n}{|S_i|} + \frac{h}{n}\log n.
\end{align*}
The first term in this sum is maximized when the non-participating points are fairly divided between the $k$ different sets $S_i$ (by the concavity of the considered function). Moreover, because $k \leq \min\{n-h,(h+1)^2\}$, we have $\frac{n-h}{kn}\log \frac{nk}{n-h} \in O(\frac{1}{k}\log h)$. Thus
\begin{align*}
    \entropy(\Pi) &\leq \sum_{i=1}^k \frac{(n-h)/k}{n}\log \frac{n}{(n-h)/k} + \frac{h}{n}\log n\\
    &\in O\left(\sum_{i=1}^k \frac{1}{k}\log h + \frac{h}{n}\log n\right) \\
    &= O(\log h + \frac{h}{n}\log n).
\end{align*}

Because $h\log n \in O(n\log h)$ we finally get
\begin{align*}
    n(\entropy(\Pi)+1) &\in O(n\log h + h\log n + n) \\
    & \subset O(n\log h).
\end{align*}
This concludes the proof.
\end{proof}

\section{Instance-optimality is impossible with linear queries}\label{sec:impossibility_linear_queries}

In the previous section, we have shown that there is a comparison-based algorithm to report participating points which is instance-optimal in the order-oblivious runtime against all comparision-based algorithm solving the problem. We also show that if we ``compete'' against algorithms which can do queries of the form $x(p)-x(q) \geq y(p)-y(q)$, then such a result is no longer possible, even if we allow our algorithm to be in a much stronger model of computation. 

\begin{figure}
    \centering
    \includegraphics[scale=0.5]{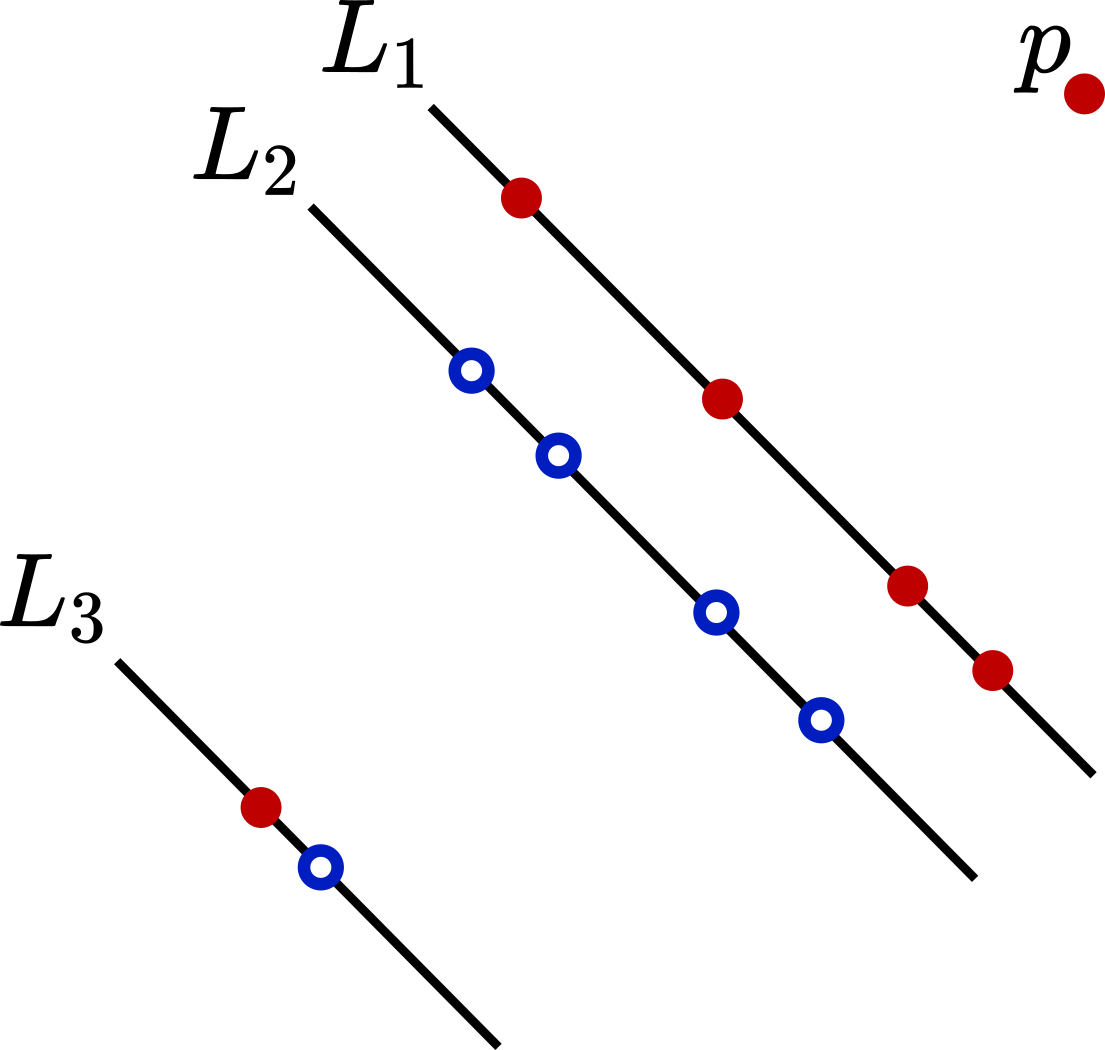}
    \caption{Example of an instance in $D$.}
    \label{fig:impossibility-instance}
\end{figure}

\begin{proposition}
Let $\A$ be the class of algorithms in the real RAM model which only access the input through comparison of coordinates and queries of the form $x(p)-x(q) \geq y(p)-y(q)$.
For the problem of reporting participating points in the order-oblivious setting, there is no algebraic computation tree $A$ such that $T_A(S) \leq O(1)T_{A'}(S)$ for every instance $S$ and every correct algorithm $A'\in \A$.

In other words, no algebraic computation tree can perform within a constant factor of complexity against every (correct) algorithm in $\A$ on every input.
\end{proposition}
\begin{proof}

Consider the set of instances $D$ with the following structure:
We have three lines $L_1$, $L_2$, $L_3$ with slope $-1$ where $L_1$ is above $L_2$ and $L_2$ is above $L_3$. There are $n$ distinct red points on $L_1$, $n$ distinct blue points on $L_2$. On $L_3$ there is a red and a blue point such that their $x$ and $y$ coordinates are smaller than the $x$ and $y$ coordinates of any of the points on $L_1$ and $L_2$. Finally, there is an additional red point $p$ which dominates all others. See Figure \ref{fig:impossibility-instance} for an example.

This configuration can be detected in linear time by an algorithm in $\A$. It is easy to see that when an instance is in this configuration, every point needs to be reported with the possible exception of $p$. Consider an instance $I$ where $p$ needs to be reported. Then there is a blue point $q$ which is rectangular-visible from $p$. Let us call $r$ the rank of $q$ in $I$ with respect to the order of $x$-coordinates. There is also an algorithm $A_r\in\A$ which can do the following.
\begin{itemize}
    \item First, detect if the instance is in $D$ in linear time.
    \item If the instance is in $D$, find the point $q$ of rank $r$ by a linear-time selection algorithm.
    \item Test in linear time if $q$ is blue and $p$ is visible from $q$.
    \item If this is the case, output all points.
    \item If any of the tests above fail, return the correct output by brute-force in $O(n^3)$ time.
\end{itemize}

The algorithm $A_r$ is correct on every instance and runs in linear time on any permutation of instance $I$ (thus it runs in linear time on $I$ in the order-oblivious setting). There is such an algorithm for every instance in $D$.

On the other hand, there is no algebraic computation tree which can solve all instances in $D$ where $p$ needs to be reported in linear time. To see this, we can relax the problem to "Is $p$ rectangular visible from some blue point?" and restrict the inputs even more to a set $D'\subset D$. We now restrict inputs so that all the red points are fixed in some position and all "corners" of the staircase formed by the points in $L_1$ lie under $L_2$.

\begin{center}
    \includegraphics[scale=0.5]{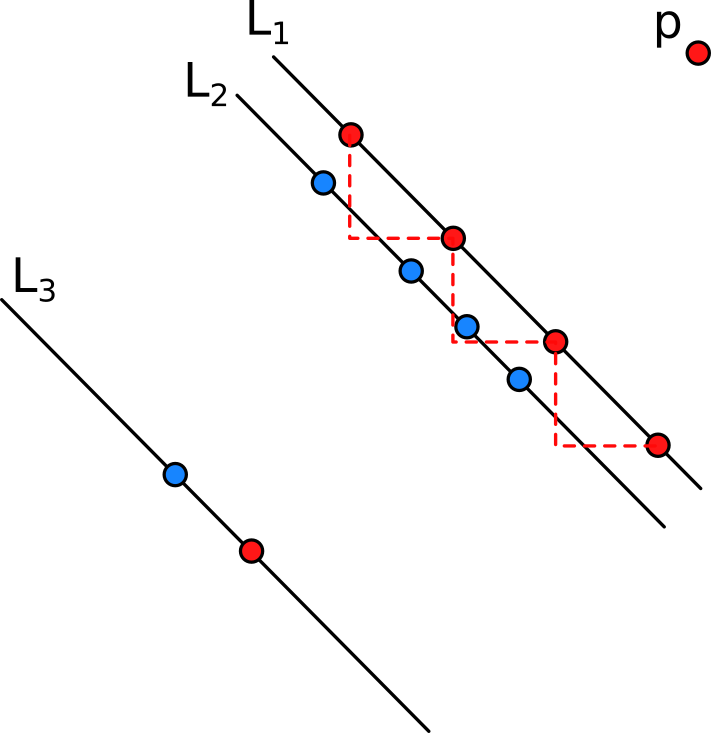}\\
    Example of an instance in $D'$.
\end{center}

It is easy to see that in the setting of the topological lower-bounds due to Ben-Or \cite{Ben-Or-1983} the ``No'' instances fall into at least $n^n$ connected components (characterized by which point on $L_1$ dominates each of the points on $L_2$). This implies a $\Omega(n\log n)$ worst case complexity for any algebraic computation tree solving all instances in $D'$. This also implies that no such computation tree solving this problem on $D'$ can run in $f(n) \not\in \Omega(n\log n)$ time on all ``Yes'' instances in $D'$. Otherwise we could this algorithm and stop the computation after $f(n)$ steps to get an algorithm with $O(f(n))$ worst-case complexity. Thus the claim holds.
\end{proof}
One caveat of this proof is that it relies on special instances with linear degeneracies (three points can be collinear). It is not clear if instance-optimality is possible when restricted to non-degenerate instances.


\newpage
\bibliography{refs}

\end{document}